\documentclass[10pt]{amsart}
\usepackage{float}
\usepackage[dvipdfmx]{graphicx,color}
\usepackage{color,colordvi}
\usepackage{url}
\usepackage{enumerate}
\newtheorem{theorem}{Theorem}[section]
\newtheorem{lemma}[theorem]{Lemma}

\theoremstyle{definition}

\theoremstyle{remark}

\newcommand{\HOX}[1]{\marginpar{\footnotesize #1}}% REPLACE BY \newcommand{\HOX}[1]{}

\newcommand{\pe}[1]{{\color{magenta}{#1}}}

\numberwithin{equation}{section}

\title[Quantum graph associated with square and hexagonal lattices
]
{Inverse problems for  quantum graph associated with square and hexagonal lattices}

\author{\hskip 3mm K. Ando, E. Bl{\aa}sten, P. Exner, H. Isozaki,\\
E. Korotyaev, M. Lassas,  J. Lu and H. Morioka}
\date{}

\AtEndDocument{\bigskip{\footnotesize%
  \textsc{Kazunori Ando: Department of Electrical and Electronic Engineering and Computer Science, Ehime University, Matsuyama, 790-8577, Japan} \par
  \textit{Email address}: \texttt{ando@cs.ehime-u.ac.jp} \par

  \addvspace{\medskipamount}

  \textsc{Emilia Bl{\aa}sten: Computational Engineering, School of Engineering Science, LUT University, Lahti campus, 15210 Lahti, Finland} \par
  \textit{Email address}: \texttt{emilia.blasten@iki.fi} \par

  \addvspace{\medskipamount}

   \textsc{Evgeny Korotyaev: Department of  Math. Analysis,
Saint-Petersburg State University,  Universitetskaya nab. 7/9, St. Petersburg, 199034, Russia,
National Research University Higher School of Economics, St. Petersburg, Russia} \par
  \textit{Email address}: \texttt{e.korotyaev@spbu.ru} \par

  \addvspace{\medskipamount}
    \textsc{Pavel Exner: Doppler Institute for Mathematical Physics and Applied Mathematics, Czech Technical University in Prague, B\v{r}ehov\'a 7, 115 19 Prague, Czechia, and Nuclear Physics Institute, Czech Academy of Science, 250 68 \v{R}e\v{z}, Czechia} \par
    \textit{Email address}: \texttt{exner@ujf.cas.cz} \par

  \addvspace{\medskipamount}

  \textsc{Hiroshi Isozaki: Graduate School of Pure and Applied Sciences, Professor Emeritus, University of Tsukuba, Tsukuba, 305-8571, Japan} \par
  \textit{Email address}: \texttt{isozakih@math.tsukuba.ac.jp} \par

  \addvspace{\medskipamount}

  \textsc{Matti Lassas: Department of Mathematics and Statistics, University of Helsinki, FI-00014 Helsinki, Finland} \par
  \textit{Email address}: \texttt{matti.lassas@helsinki.fi} \par

  \addvspace{\medskipamount}
  \textsc{Jinpeng Lu: Department of Mathematics and Statistics, University of Helsinki, FI-00014 Helsinki, Finland} \par
  \textit{Email address}: \texttt{jinpeng.lu@helsinki.fi} \par

  \addvspace{\medskipamount}
  \textsc{Hisashi Morioka: Department of Electrical and Electronic Engineering and Computer Science, Ehime University, Matsuyama, 790-8577, Japan} \par
  \textit{Email address}: \texttt{morioka@cs.ehime-u.ac.jp} \par

}}
\begin{document}
\baselineskip 14pt
\maketitle
\begin{abstract}
We solve  inverse  problems from the D-N map for the quantum graph on a finite domain in a square lattice and that on a hexagonal lattice, as well as inverse scattering problems from the S-matrix for a  locally perturbed square lattice and  a hexagonal lattice.
\end{abstract}

\section{Introduction}
\subsection{Gel'fand problem}
\label{Gel'fand problem}
In the International Congress of Mathematics at Amsterdam in 1954, I.M. Gel'fand raised the following problem (an extended form): Let $(M,g)$ be a compact Riemannian manifold with boundary. Let $\lambda_0 < \lambda_1 \leq \cdots$ be the eigenvalues and $\varphi_0(x), \varphi_1(x), \dots$  the associated orthnormal eigenvectors for the operator $ H = - c(x)^2\Delta_g + V(x)$ on $M$ with the Dirichlet boundary condition on $\partial M$.
Then, from the knowledge of the boundary spectral data (BSD)  $\big\{(\lambda_n,\frac{\partial}{\partial\nu}\varphi_n\big|_{\partial M})\, ; \, n = 0, 1, 2, \dots\big\}$, where $\nu$ is the unit normal at the boundary, can one determine the manifold $M$ and the operator $H$? \\

\indent
For the case of Riemannian manifold, this problem was solved by Belishev-Kurylev \cite{BelKur92} using the boundary control method developped  by Belishev \cite{Bel87}.  The solution is unique up to a diffeomorphsim leaving $\partial M$ invariant, and this is the only obstruction for this problem.

\subsection{Quantum graph Hamiltonian and discrete operator}
\label{SubsecgraphHamiltonianedgeHamiltonian}

\indent
We are interested in the analogue of Gel'fand's problem on graphs. There are two models in which one can do that. The first one concerns graphs $\Gamma = \{\mathcal V,\mathcal E\}$ understood as a vertex set $\mathcal V$ and an edge set $\mathcal E$ determined by an adjacency matrix; we are then interested in the {\it discrete operator} defined on the vertex set,
 %----------------%	
\begin{equation}
\hat H \hat u(x) = \frac{1}{\mu_x}\sum_{y \sim x} g_{xy}(\hat u(y) - \hat u(x)) +  q(x)\hat u(x), \quad x \in \mathcal V,
\label{DiscretegraphLaplacian}
\end{equation}
 %----------------%	
where $y \sim x$ means that $x$ and $y$ are the endpoints of a same edge $e = e_{xy} \in \mathcal E$, $\mu : \mathcal V \to {\mathbb R}_+ = (0,\infty)$ is a weight on $\mathcal V$, $g : \mathcal E \to {\mathbb R}_+$ is a weight on $\mathcal E$ and $q : \mathcal V \to {\mathbb R}$ is a scalar potential. The other one concerns the so-called \emph{quantum graphs} \cite{BerkolaikoKuchment2013} in which the edges are identified with line segments and the Hamiltonian is a collection of one-dimensional Schr{\"o}dinger operators defined on them,
 %----------------%	
\begin{equation}
\widehat H_{\mathcal E} = \Big\{h_e =  - \frac{d^2}{dz^2} + V_e(z),  \ z \in [0,\ell_e] \Big\}_{e \in \mathcal E},
\end{equation}
 %----------------%	
with a real-valued potential $V_e(z)$. To make such an operator self-adjoint, one has to match the function properly at the vertices; we choose the simplest possibility assumig the $\delta$-coupling condition, see (\ref{Deltacoupling1}) below. \\
 \indent
It was proved in \cite{BILL2021}  that in the case of the discrete graph,  the graph structure and the coefficients of (\ref{DiscretegraphLaplacian}) are determined by BSD  under the Two-Point Condition introduced there. We emphasize here that the inverse problem for the discrete graph cannot be solved without such conditions, as documented by a counter example given in \cite{BILL2021}. Recall that the knowledge of BSD is equivalemt to the knowledge of the associated D-N map for all energies. The indicated result has various applications, in particular, those concerning the following three  issues:
 %----------------%	
\begin{enumerate}
\item Inverse boundary value problems for random walks.
\item
 Inverse scattering problems for locally perturbed discrete periodic graphs.
\item  Inverse boundary value problems and inverse scattering problems for quantum (metric) graph.
\end{enumerate}
 %----------------%	
The problem (1) has been discussed in \cite{BILL2}, where it was shown that the graph structure and the transition matrix of the random walk can be uniquely recovered from the distribution of the first passing time on the boundary, or from the observation on the boundary of one realization of the random walk.  \\
\indent
The problem (2) is  considered in \cite{BEILL}.
For a locally perturbed discrete periodic graph, the structure of the perturbed subgraph, along with the edge weight $g$ and the potential $q$ can be recovered from the scattering matrix at all energies provided that the Two-Points Condition is preserved under the perturbations.
%Here, we remark that  we have proven more than the statement in Theorem 5.11 in \cite{BEILL}. Although we assumed  $g  \equiv 1$ there, this assumption is made only for simplifying the statement of the theorem, and not absolutely necessary. In fact, in the spectral theory for the perturbed periodic lattices, especially in the forward problem, local perturbation by $g$ can be treated by the standard perturbation theory and plays no essential role as long as the unique continuation property holds for the Hamiltonian.
For a  fixed energy~$\lambda$, the S-matrix $S(\lambda)$ of the whole system and the D-N map $\Lambda(\lambda)$ determine each other, and thus the problem is reduced to that of the bounded domain, to which we can apply the result of \cite{BILL2021}. In particular, if two locally perturbed periodic lattices have the same S-matrix for all energies, we can conclude:
 (i) If $\mu = \mu'$, then $g = g'$ and $q = q'$. (ii) If $q = q'$, then $\mu = \mu'$ and $g = g'$. (iii) In particular, if $\mu(v) = {\rm deg}\,(v)$, $\mu'(v') = {\rm deg}\,(v')$, then $g = g'$ and $q=q'$.\\
\indent
The well known duality between the discrete and metric graphs \cite{Exner97} allows to determine the spectrum of an \emph{equilateral} `continuous' graph from that of its discrete counterpart and \emph{vice versa}, and therefore the discrete Gel'fand problem is expected to play a role in inverse problems for quantum graphs. In \cite{BEILL}, we have considered such equilateral graphs, i.e. those
in which all the edges have the same length and the one-dimensional Hamiltonian has the same potential on them.
Letting $C_v$ be the $\delta$-coupling constant (see (\ref{Deltacoupling1})  below), and $d_v$ the degree of $v \in \mathcal V$, it was assumed that $C_v/d_v$ is a constant independent of $v \in \mathcal V$. We then proved that if two such quantum graphs ${\mathbb G}_{\Gamma}$ and ${\mathbb G}_{\Gamma'}$ have the same D-N map (or the same S-matrix) for all energies, then there exists a bijection ${\mathbb G}_{\Gamma} \to {\mathbb G}_{\Gamma'}$ preserving the edge relation. Moreover,  we have $d_v = d_{v'}$ and $C_v = C_{v'}$ $\forall v \in \mathcal V$.
Therefore, the S-matrix of an equilateral quantum graph determines its graph structure. \\
\indent
In this paper, we continue the investigation of \cite{BEILL},  and consider the problem of determining local perturbation of $C_v$ and $V_e(z)$ for lattices such that $\ell_e = 1$ holds for all $e \in \mathcal E$ and the degree $d_v$ is the same for all $v \in \mathcal V$.
Our method is different from that of \cite{BEILL} relying strongly on the explicit form of the lattice. Therefore, although we are convinced that the result will be true for a larger class of lattices, we restrict ourselves to the proof  for  the square and hexagonal situations.

\subsection{Edge and vertex  Schr\"odinger operators in the quantum graph}

Let us  recall basic facts about quantum graphs. Let  $\Gamma = \{\mathcal V, \mathcal E\}$ be a quantum (or metric) graph with a vertex set $\mathcal V$ and edge set $\mathcal E$. We assume that  each edge $e$ has unit length and identify it with the interval $[0,1]$. Consider a quantum-graph  Schr\"odinger operator, called the {\it  edge Schr{\"o}dinger operator} in this paper,
 %----------------%	
\begin{equation}
\widehat H_{\mathcal E} = \left\{h_e =  - \frac{d^2}{dz^2} + V_e(z),  \ z \in [0,1] \right\}_{e \in \mathcal E},
\end{equation}
 %----------------%	
with a real-valued potential satisfying the symmetry condition
 %----------------%	
\begin{equation}
V_e(z) \in L^2((0,1)), \quad V_e(z) = V_e(1 -z), \quad \forall e \in \mathcal E.
\end{equation}
 %----------------%	
The generalized Kirchhoff condition, or the {\it $\delta$-coupling condition}, is imposed\footnote{Here,  $v \in e$ means that $v$ is an end point of an edge $e$ and the derivative is conventionally taken in the outward direction.}:
 %----------------%	
\begin{equation}
\sum_{v \in e}\widehat u'_e(v) = C_v\widehat u_e(v), \quad \forall v \in \mathcal V^{o} := \mathcal V\setminus \partial{\mathcal V},
\label{Deltacoupling1}
\end{equation}
 %----------------%	
$C_v$ being a real constant. Here, $\partial\mathcal V$ is the boundary of $\mathcal V$, which will be chosen suitably later. Let $\phi_e(z,\lambda)$ be the solution of
 %----------------%	
\begin{equation}
\left\{
\begin{split}
& \Big(- \frac{d^2}{dz^2} + V_e(z)\Big)\phi_e(z,\lambda) = \lambda \phi_e(z,\lambda), \quad
z \in [0,1], \\
& \phi_e(0,\lambda) = 0, \quad \phi'_e(0,\lambda) = 1, \quad  ^{\prime}= \frac{d}{dz}.
\end{split}
\right.
\end{equation}
 %----------------%	
We put
 %----------------%	
\begin{equation}
\phi_{e0}(z,\lambda) = \phi_e(z,\lambda), \quad
\phi_{e1}(z,\lambda) = \phi_{e}(1 - z,\lambda).
\label{Definephie0}
\end{equation}
 %----------------%	
Each edge $e$ is parametrized as $e(z)$, $0 \leq z \leq 1$. Letting $e(0) = v$, $e(1) = w$, any solution $u = \{u_e(z,\lambda)\}_{e \in \mathcal E}$ of the equation $(h_e - \lambda)\widehat u = 0$ can be written as
 %----------------%	
\begin{equation}
\widehat u_e(z,\lambda) = \widehat u(v,\lambda)\frac{\phi_{e0}(z,\lambda)}{\phi_{e0}(1,\lambda)} +  \widehat u(w,\lambda)\frac{\phi_{e1}(z,\lambda)}{\phi_{e1}(1,\lambda)}.
\end{equation}
 %----------------%	
The edge Schr{\"o}dinger operator $\widehat H_{\mathcal E}$ is related to the {\it vertex  Schr{\"o}dinger operator} in the following way: we define the operators $\widehat\Delta_{\mathcal V,\lambda}$ and $\widehat Q_{\mathcal V,\lambda}$ on $\mathcal V$ by
 %----------------%	
\begin{equation}
\left(\widehat\Delta_{\mathcal V,\lambda}\widehat u\right)(v) = \frac{1}{d_v}\sum_{w \sim v}\frac{1}{\phi_{e0}(w,\lambda)}\widehat u(w,\lambda),
\label{VertexSchrodingerop}
\end{equation}
 %----------------%	
\begin{equation}
\widehat Q_{\mathcal V,\lambda} = \frac{1}{d_v}\sum_{v \in e}
\frac{\phi'_{e0}(1,\lambda)}{\phi_{e0}(1,\lambda)} + \frac{C_v}{d_v},
\label{QVlambdaformula}
\end{equation}
 %----------------%	
where $d_v$ is the degree of $v \in \mathcal V$. Then, the $\delta$-coupling condition (\ref{Deltacoupling1}) is rewritten in the form of vertex Schr{\"o}dinger equation
 %----------------%	
\begin{equation}
\left(- \widehat\Delta_{\mathcal V,\lambda} + \widehat Q_{\mathcal V,\lambda}\right)\widehat u(v) = 0, \quad \forall v \in \mathcal V^{0}.
\end{equation}
 %----------------%	
%In \cite{BEILL} (including the previous works \cite{AIM16}, \cite{AIM18}), we imposed  groups of assumptions
 %(M-1)-(M-5),  (A-1)-(A-4), (B-1)-(B-3), (C-1), (C-2), (C-1)', (D-1)-(D-4) and (E-1) to deal with these operators.

In our previous work \cite{BEILL}, we studied spectral properties of quantum graphs on a class of  locally (i.e. on a bounded part) perturbed periodic lattices including the square and hexagonal ones which are the object of interest in this paper. We have found, in particular,
 %The assumptions M's, A's and B's are related to the lattice structure and the $\delta$-coupling condition, and are satisfied for the square and hexagonal lattices.
 %The other conditions C's, D's and (E-1) are related to the perturbations in a finite part of the lattice, which are satisfied automatically in our case by the following reason. Under these assumptions, we have proven
that the N-D map (and the D-N map as well) for the vertex Schr\"odinger operator on the interior domain and the S-matrix  for the edge Schr\"odinger operator on the whole system determine each other, cf.~Corollary 6.16 in \cite{BEILL}. Therefore, one can  reduce the inverse boundary value problem for these lattices to that on a domain of the shape to be described below, namely the rectangular domain for the square lattice and the hexagonal parallelogram for the hexagonal lattice.

\subsection{Inverse boundary value problem for square and hexagonal lattices}

Let $\Gamma_0 = \{\mathcal V_0, \mathcal E_0\}$ be the square or the hexagonal lattice in ${\mathbb R}^2$ with vertex set $\mathcal V_0$ and edge set $\mathcal E_0$. Assume that we are given a bounded domain $\Omega$ in $\mathbb R^2$ and a subgraph $\Gamma = \{\mathcal V, \mathcal  E\}$, where $\mathcal E = \mathcal E_0 \cap \Omega$, $\mathcal V = \mathcal V_0 \cap \Omega$. We define $\partial \mathcal E$ to be the set of $v \in \mathcal V$ such that $v$ is an end point of some edge in $\mathcal E$, and ${\rm deg}_{\Gamma}(v) = 1$, where ${\rm deg}_{\Gamma}(v)$ is the degree of $v$ in the graph $\Gamma$. We further assume that $\Gamma$ has the following properties\footnote{These assumptions were used in \cite{BEILL}; we state them again to make the paper self-contained.}:

\medskip
\noindent
{\it
(i)  $\ \partial \mathcal E = \partial \mathcal V.$ 

\smallskip
\noindent
(ii) $d_v = 1\ { for} \ v \in \partial \mathcal V.$

\smallskip
\noindent
(iii) The unique  continuation  property  holds  for  the vertex Schr\"odinger  operator  in  the  exterior  domain $ \mathcal V_{ext} := \mathcal V_0 \setminus \mathcal V $  in the following sense: if $ \widehat u$ satisfies the equation $(- \widehat\Delta_{\mathcal V,\lambda} + \widehat Q_{\mathcal V,\lambda})\widehat u = 0$ in $ \mathcal V_{ext}$ and $ \widehat u = 0$  in $\{v \in \mathcal V_{ext}\, ; \, |v| > R\}$ for some $ R > 0$,  then $\widehat u = 0$ on $\mathcal V_{ext}.$}

\medskip
 %Moreover,  assume that we know $V_e(z)$ for $e$ adjacent to $\partial \mathcal V$.
The D-N map for the edge  Schr\"odinger operator is defined by
 %----------------%	
\begin{equation}
\Lambda_{\mathcal E}(\lambda) : \widehat f \to \widehat u_e'(v), \quad
e(0) = v \in \partial\mathcal V,
\end{equation}
 %----------------%	
where $\widehat u_e$ is the solution to the equation
 %----------------%	
\begin{equation}
\left\{
\begin{split}
& (h_e - \lambda)\widehat u_e = 0, \quad on \quad \mathcal E,\\
& \widehat u_e = \widehat f, \quad on \quad \partial\mathcal E = \partial\mathcal V, \\
&\delta{-}coupling \ condition.
\end{split}
\right.
\end{equation}
 %----------------%	
The D-N map for the vertex  Schr\"odinger operator is defined by
 %----------------%	
\begin{equation}
\Lambda_{\mathcal V}(\lambda) : \widehat f(v) \to \frac{1}{\phi_{e0}(1,\lambda)}\widehat u(w),  \
v = e(0) \in \partial\mathcal V, \ w = e(1) \in \mathcal V^{o}.
\label{DefineDNlapforvertexoperator}
\end{equation}
 %----------------%	
where $\widehat u$ is the solution to the equation
 %----------------%	
\begin{equation}
\left\{
\begin{split}
& \left(- \widehat\Delta_{\mathcal V,\lambda} + \widehat Q_{\mathcal V,\lambda}\right)\widehat u(v) = 0, \quad v \in \mathcal V^{0}, \\
& \widehat u(v) = \widehat f(v), \quad v \in \partial\mathcal V.
\end{split}
\right.
\end{equation}
 %----------------%	
Then, by Lemma 3.1 in \cite{BEILL},  $\Lambda_{\mathcal E}(\lambda)$ and $\Lambda_{\mathcal V}(\lambda)$ determine each other under the assumptions (i), (ii), (iii).

\medskip

Our first main result is as follows:
 %----------------%	
\begin{theorem}
\label{TheoremIBVP}
Let $\Omega$ be  a bounded domain in the 2-dimensional square or hexagonal lattice having the properties (i), (ii), (iii), and consider the edge  Schr\"odinger operator $\widehat H_{\mathcal E} = \{- d^2/dz^2 + V_e(z)\}_{e \in \mathcal E}$ assuming the $\delta$-coupling condition and the Dirichlet boundary condition on $\partial \mathcal E$.  Then one can uniquely reconstruct $V_e(z)$ and $C_v$ for all $e \in \mathcal E$ and $v \in \mathcal V$
from the D-N map of $\widehat H_{\mathcal E}- \lambda$ for all values of the energy $\lambda$, provided we know $V_e(z)$ for all the edges $e$ adjacent to the boundary of $\mathcal V$.
\end{theorem}
 %----------------%	

\subsection{Inverse scattering problem}

In \cite{BEILL}, we have discussed the spectral and scattering theory for Schr{\"o}dinger operators on a class of locally perturbed periodic quantum graphs. Here we will use these results for perturbations of square and hexagonal lattices quantum graphs preserving the lattice structure. Assume that the length of  each edge is one. The assumptions on $V_e(z)$ and $C_v$ are as follows:

\smallskip
\noindent
{\it
(iv)  There exists a constant $C_0 \in \mathbb R$ such that $C_v = C_0$ except for a finite number of vertices $v \in \mathcal V$.

\smallskip
\noindent
(v)  There exists $V_0 \in L^2((0,1))$ such that $V_e(z)  = V_0(z)$ except for a finite number of edges $e \in \mathcal E$.}

\medskip
One can then define the S-matrix $S(\lambda)$ for the Hamiltonian of the quantum graph built on $\Gamma$. Consider a bounded domain  $\Omega$ in $\mathcal E$
which contains all the indicated perturbations, in particular, assume that $V_e(z) = V_0(z)$ holds on any edge $e$ adjacent to $\partial\mathcal V$.
In Corollary 6.16 in \cite{BEILL}, we have proven that the S-matrix $S(\lambda)$ and the D-N map $\Lambda_{\mathcal E}(\lambda)$ for $\mathcal V$ determine each other.
Applying Theorem \ref{TheoremIBVP}, we obtain our second main result:

 %----------------%	
\begin{theorem}
\label{TheoremSquareISP}
Consider the Schr\"odinger operator $\widehat H_{\mathcal E} = \{- d^2/dz^2 + V_e(z)\}_{e \in \mathcal E}$ on the 2-dimensional square or hexagonal lattice $\Gamma$ satisfying the conditions (i)-(v). Then one can uniquely reconstruct $V_e(z)$ and $C_v$ for all $e \in \mathcal E$ and $v \in \mathcal V$ from the knowledge of the S-matrix $S(\lambda)$ of $\widehat H_{\mathcal E}$ for all energies $\lambda$.
\end{theorem}
 %----------------%	

Let us remark that by assumption we know $V_0(z)$ and $C_0$ a priori.
 %As will be seen in the proof below, our proof depends strongly on the shape of the lattice.
%Therefore, the extension of the above theorems to other graphs is not an obvious problem.
%Since the S-matrix and the
%D-N map are equivalent, to prove Theorems \ref{TheoremIBVP} and \ref{TheoremSquareISP}, we have only to show them for some particular domain. For the case of square lattice, we choose a rectangular domain as in Figure \ref{SquareIntDomain}, and for the case of hexagonal lattice, we deal with a hexagonal parallelogram as in Figure \ref{S6HexaParallel}.

\subsection{Related works}
The study of spectral and scattering theory on quantum graphs is now making a rapid progress. Although  the topics of this paper are restricted to inverse problems on quantum graphs for square and hexagonal lattices,
there are plenty of articles devoted to this subject. Our previous work \cite{BILL2, BEILL}, on which the present paper is based, deals with the forward problem and some types of inverse problems. We add here a brief look at the related works.

A general survey of discrete graph and quantum graph properties can be found in the monographs \cite{BerkolaikoKuchment2013, KN22, Ku24, Post12}, see also the paper \cite{EKMN17}.
For dicussion of the $\delta$-coupling and related topics we refer to \cite{ChExTu10, Exner96, Exner97, ExnerKovarik2015, ExnerPost09}. The relation between edge Schr\"odinger operators and vertex Schr\"odinger operators was studied in \cite{Below85, BPG08, Cattaneo97, Exner97, KostrykinSchrader1999, Pankrashkin13}. Spectral properties of quantum graphs are discussed in \cite{KoLo07, KoSa14, KoSa15, KoSa15b, Niikuni16, Pank06, ParRich18}. The wave operators for discrete Schr{\"o}dinger operators are investigated in \cite{Nakamura14, ParRich18, Tadano19}. Various inverse problems for the quantum graphs are studied in \cite{FrYu01}, see also \cite{GutSmil01, PoTru}. Finally, for earlier results on the inverse scattering for discrete Schr\"odinger operators on locally perturbed periodic graphs see \cite{A1, AIM16, AIM18, AvdBelMat11, Bel04, BILL2021, Pivo00, VisComMirSor11, YangXu18, Yurk05, Yurko16(1)}.

\subsection{Acknowledgement}
The authors express their gratitude for the funding obtained. P.E. was supported  by the EU under the Marie Sk{\l}odowska-Curie Grant No 873071. H.I. was supported by Grant-in-Aid for Scientific Research (C) 20K03667 and (C) 24K06768 Japan Society for the Promotion of Science. H.M. was supported by Grant-in-aid for young scientists 20K14327 Japan Society for the Promotion of Science. The work of E.B. was supported by the Research Council of Finland
through the Flagship of Advanced Mathematics for Sensing, Imaging and
Modelling (decision number 359183).

\section{Square lattice}
\label{SectionSquareLattice}
As we have proven in Theorem 5.7 of \cite{BEILL}, the S-matrix and the D-N map determine each other, if the unique continuation theorem holds in the exterior domain. Then, we can change the domain $\Omega$ as long as the conditions (i), (ii), (iii) hold. Therefore, to prove  Theorem \ref{TheoremIBVP}, we have only to consider the case in which $\Omega$ is a rectangular domain as below.

Given a square lattice $\Gamma_0 = \{\mathcal V_0, \mathcal E_0\}$ in ${\mathbb R}^2$, let ${\Omega}$ be its rectangular domain as sketched in Figure~\ref{SquareIntDomain}, and $\Gamma = \{\mathcal V, \mathcal E\}$, where $\mathcal V = \mathcal V_0\cap \Omega$, $\mathcal E = \mathcal E_0 \cap \Omega$. The black dots  there denote the boundary points satisfying $d_v = 1$ for $v \in \partial \mathcal V$, while $d_v = 4$ for $v \in \mathcal V^{o}$. The  boundary $\partial \mathcal V$ consists of four parts $(\partial \mathcal V)_T$, $(\partial \mathcal V)_B$, $(\partial \mathcal V)_L$, $(\partial \mathcal V)_R$, where the top $(\partial \mathcal V)_T$ and the left side $(\partial \mathcal V)_L$ are given by
 %----------------%	
\begin{equation}
(\partial \mathcal V)_T = \{a_1, a_2, \dots,a_m\},\quad (\partial \mathcal V)_L = \{b_1, b_2,\dots, b_n\},
\end{equation}
 %----------------%	
and the bottom $(\partial \mathcal V)_B$ and the right side $(\partial \mathcal V)_R$ are defined similarly.
 %----------------%	
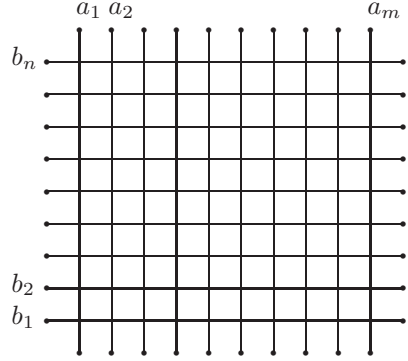
\begin{figure}[hbtp]
\begin{center}
\setlength{\unitlength}{0.43cm}
\begin{picture}(11,12)
\put(0,1){\line(1,0){11}}
\put(0,2){\line(1,0){11}}
\put(0,3){\line(1,0){11}}
\put(0,4){\line(1,0){11}}
\put(0,5){\line(1,0){11}}
\put(0,6){\line(1,0){11}}
\put(0,7){\line(1,0){11}}
\put(0,8){\line(1,0){11}}
\put(0,9){\line(1,0){11}}

\put(1,0){\line(0,1){10}}
\put(2,0){\line(0,1){10}}
\put(3,0){\line(0,1){10}}
\put(4,0){\line(0,1){10}}
\put(5,0){\line(0,1){10}}
\put(6,0){\line(0,1){10}}
\put(7,0){\line(0,1){10}}
\put(8,0){\line(0,1){10}}
\put(9,0){\line(0,1){10}}
\put(10,0){\line(0,1){10}}

\put(1,0){\circle*{0,2}}
\put(2,0){\circle*{0,2}}
\put(3,0){\circle*{0,2}}
\put(4,0){\circle*{0,2}}
\put(5,0){\circle*{0,2}}
\put(6,0){\circle*{0,2}}
\put(7,0){\circle*{0,2}}
\put(8,0){\circle*{0,2}}
\put(9,0){\circle*{0,2}}
\put(10,0){\circle*{0,2}}

\put(1,10){\circle*{0,2}}
\put(2,10){\circle*{0,2}}
\put(3,10){\circle*{0,2}}
\put(4,10){\circle*{0,2}}
\put(5,10){\circle*{0,2}}
\put(6,10){\circle*{0,2}}
\put(7,10){\circle*{0,2}}
\put(8,10){\circle*{0,2}}
\put(9,10){\circle*{0,2}}
\put(10,10){\circle*{0,2}}

\put(0,1){\circle*{0,2}}
\put(0,2){\circle*{0,2}}
\put(0,3){\circle*{0,2}}
\put(0,4){\circle*{0,2}}
\put(0,5){\circle*{0,2}}
\put(0,6){\circle*{0,2}}
\put(0,7){\circle*{0,2}}
\put(0,8){\circle*{0,2}}
\put(0,9){\circle*{0,2}}

\put(11,1){\circle*{0,2}}
\put(11,2){\circle*{0,2}}
\put(11,3){\circle*{0,2}}
\put(11,4){\circle*{0,2}}
\put(11,5){\circle*{0,2}}
\put(11,6){\circle*{0,2}}
\put(11,7){\circle*{0,2}}
\put(11,8){\circle*{0,2}}
\put(11,9){\circle*{0,2}}

\put(0.9,10.5){$a_1$}
\put(1.9,10.5){$a_2$}
\put(9.9,10.5){$a_m$}

\put(-1.1,0.9){$b_1$}
\put(-1.1,1.9){$b_2$}
\put(-1.1,8.9){$b_n$}

\end{picture}
\end{center}
\caption{Rectangular domain}
\label{SquareIntDomain}
\end{figure}
 %----------------%	

\noindent Let  $- \widehat\Delta_{\mathcal V, \lambda} + \widehat Q_{\mathcal V,\lambda}$ be the vertex Hamiltonian introduced in Subsection \ref{SubsecgraphHamiltonianedgeHamiltonian}.
By {\it a~cross} with center $v_0$ we mean the graph shown in Figure~\ref{Cross}.
 %----------------%	
\begin{figure}[hbtp]
\begin{center}
\setlength{\unitlength}{0.8cm}
\begin{picture}(9,5)
\put(5, 4.5){\circle*{0,2}}
\put(5,2.5){\circle*{0,2}}
\put(5,0.5){\circle*{0,2}}
\put(3,2.5){\circle*{0,2}}
\put(7,2.5){\circle*{0,2}}
\put(3,2.5){\line(1,0){4}}
\put(5,0.5){\line(0,1){4}}
\put(5.2,2.7){$v_0$}
\put(7.2,2.5){$v_1$}
\put(4.8,4.8){$v_2$}
\put(2.4,2.5){$v_3$}
\put(4.8,0){$v_4$}
\end{picture}
\end{center}
\caption{Cross}
\label{Cross}
\end{figure}
 %----------------%	
Denoting by $e_i$ the edge with the endpoints $v_0$ and $v_i$, we can rewrite the equation $(- \widehat \Delta_{\mathcal V,\lambda} + \widehat Q_{\mathcal V,\lambda})\widehat u = 0$ as
 %----------------%	
\begin{equation}
\sum_{i=1}^4\frac{1}{\phi_{e_i}(1,\lambda)} \widehat u(v_i) =
\Big(\sum_{i=1}^4\frac{\phi'_{e_i}(1,\lambda)}{\phi_{e_i}(1,\lambda)} + C_{v_0}\Big)\widehat u(v_0).
\label{Equationwithcenterv0}
\end{equation}
 %----------------%	
The key to the inverse procedure is the following partial data problem \cite[Lemma~6.1]{AIM18}. Denoting $\mathcal V^{o} = \mathcal V\setminus{\partial\mathcal V}$, we define Neumann derivative\footnote{Note that this definition of Neumann derivative differs from that in \cite{BEILL}; here we adopt the one employed in \cite{AIM18}.} on the boundary for the vertex Hamiltonian by
 %----------------%	
\begin{equation}
\big(\partial_{\nu}\hat u\big)(v) = -  \frac{1}{d_v}\sum_{w \sim v, w \in \mathcal V^o}\frac{1}{\phi_{e0}(w,\lambda)}\widehat u(w), \quad v \in \partial\mathcal V,
\end{equation}
 %----------------%	
where $\phi_{e0}(w,\lambda)$ is given in (\ref{Definephie0}).

%%%%%%%%%% Lemma 2.1 %%%%%%%%%

 %----------------%	
\begin{lemma}\label{S6partialDNdata}
(1) Given partial Dirichlet data $\widehat f$ on $\partial\mathcal V\setminus(\partial \mathcal V)_R$, and partial Neumann data $\widehat g$ on $(\partial\mathcal V)_L$, there is a unique solution $\widehat u$ on $\mathcal V$ to the equation
 %----------------%	
\begin{equation}
\left\{
\begin{split}
& (- \widehat\Delta_{\mathcal V,\lambda} + \widehat Q_{\mathcal V,\lambda})\widehat u = 0 \quad {\it in}\; \mathcal V^{o},\\
& \widehat u =\widehat f \quad {\it on}\; \partial\mathcal V\setminus(\partial\mathcal V)_R, \\
& \partial_{\nu}\widehat u = \widehat g \quad {\it on} \; (\partial\mathcal V)_L.
\end{split}
\right.
\label{Lemma61Equation}
\end{equation}
 %----------------%	
\noindent
(2) Given the D-N map $\Lambda_{\mathcal V}(\lambda)$, partial Dirichlet data $\widehat f_2$ on $\partial\mathcal V\setminus(\partial\mathcal V)_R$ and partial Neumann data $\widehat g$ on $(\partial\mathcal V)_{L}$, there exists a unique $\widehat f$ on $\partial\mathcal V$ such that $\widehat f = \widehat f_2$ on $\partial\mathcal V\setminus(\partial\mathcal V)_R$ and $\Lambda_{\mathcal V}(\lambda)\widehat f = \widehat g$ on $(\partial\mathcal V)_{L}$. Moreover, $\widehat f$ is uniquely determined by the D-N map.
\end{lemma}
 %----------------%	

Let $A_k$ be the line with the slope $-1$ passing through $a_k$ as sketched in Figure~\ref{AkAk-1Square}. Denote the vertices on $A_k \cap \mathcal V$ by
 %----------------%	 	
$$
a_k = a_{k,0},\ \  a_{k,1},\ \ \dots\ \ ,\ \  a_k^{\ast},
$$
 %----------------%	
successively. Then, $A_k \cap \partial\mathcal V = \{a_k, a_k^{\ast}\}$.
 %----------------%	
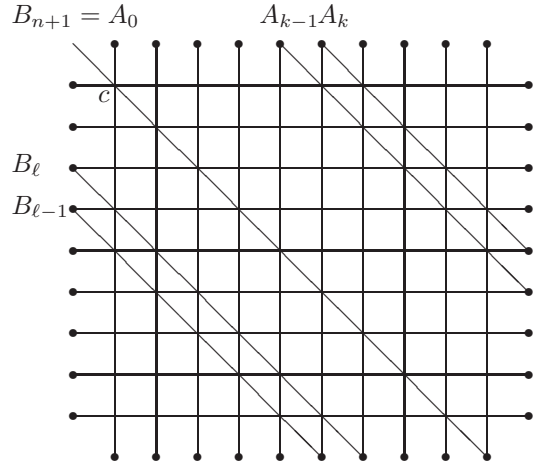
\begin{figure}[hbtp]
\begin{center}
\setlength{\unitlength}{0.55cm}
\begin{picture}(11,12)
\put(0,1){\line(1,0){11}}
\put(0,2){\line(1,0){11}}
\put(0,3){\line(1,0){11}}
\put(0,4){\line(1,0){11}}
\put(0,5){\line(1,0){11}}
\put(0,6){\line(1,0){11}}
\put(0,7){\line(1,0){11}}
\put(0,8){\line(1,0){11}}
\put(0,9){\line(1,0){11}}
\put(0.6,8.6){$c$}

\put(1,0){\line(0,1){10}}
\put(2,0){\line(0,1){10}}
\put(3,0){\line(0,1){10}}
\put(4,0){\line(0,1){10}}
\put(5,0){\line(0,1){10}}
\put(6,0){\line(0,1){10}}
\put(7,0){\line(0,1){10}}
\put(8,0){\line(0,1){10}}
\put(9,0){\line(0,1){10}}
\put(10,0){\line(0,1){10}}

\put(1,0){\circle*{0,2}}
\put(2,0){\circle*{0,2}}
\put(3,0){\circle*{0,2}}
\put(4,0){\circle*{0,2}}
\put(5,0){\circle*{0,2}}
\put(6,0){\circle*{0,2}}
\put(7,0){\circle*{0,2}}
\put(8,0){\circle*{0,2}}
\put(9,0){\circle*{0,2}}
\put(10,0){\circle*{0,2}}

\put(1,10){\circle*{0,2}}
\put(2,10){\circle*{0,2}}
\put(3,10){\circle*{0,2}}
\put(4,10){\circle*{0,2}}
\put(5,10){\circle*{0,2}}
\put(6,10){\circle*{0,2}}
\put(7,10){\circle*{0,2}}
\put(8,10){\circle*{0,2}}
\put(9,10){\circle*{0,2}}
\put(10,10){\circle*{0,2}}

\put(0,1){\circle*{0,2}}
\put(0,2){\circle*{0,2}}
\put(0,3){\circle*{0,2}}
\put(0,4){\circle*{0,2}}
\put(0,5){\circle*{0,2}}
\put(0,6){\circle*{0,2}}
\put(0,7){\circle*{0,2}}
\put(0,8){\circle*{0,2}}
\put(0,9){\circle*{0,2}}

\put(11,1){\circle*{0,2}}
\put(11,2){\circle*{0,2}}
\put(11,3){\circle*{0,2}}
\put(11,4){\circle*{0,2}}
\put(11,5){\circle*{0,2}}
\put(11,6){\circle*{0,2}}
\put(11,7){\circle*{0,2}}
\put(11,8){\circle*{0,2}}
\put(11,9){\circle*{0,2}}

\put(0,10){\line(1,-1){10}}
\put(-1.5,10.5){$B_{n+1} = A_0$}

\put(6,10){\line(1,-1){5}}
\put(5,10){\line(1,-1){6}}
\put(5.9,10.5){$A_k$}
\put(4.5,10.5){$A_{k-1}$}

\put(0,7){\line(1,-1){7}}
\put(0,6){\line(1,-1){6}}
\put(-1.5,6.9){$B_{\ell}$}
\put(-1.5,5.9){$B_{\ell-1}$}

\end{picture}
\end{center}
\caption{Lines $A_k$ and $B_{\ell}$}
\label{AkAk-1Square}
\end{figure}
 %----------------%	

\begin{lemma}
\label{Lemmaspecialboundarydata}
(1) There exists a unique solution $\widehat u$ to the equation
 %---------------%
\begin{equation}
\big(- \widehat\Delta_{\mathcal V,\lambda} + \widehat Q_{\mathcal V,\lambda}\big)\widehat u = 0 \quad {\it in} \; \mathcal V^{o},
\label{S2BVP}
\end{equation}
 %---------------%
with the partial Dirichlet data  $\widehat f$  \textcolor{blue}{on $\partial\mathcal V\setminus(\partial\mathcal V)_R$} 
\HOX{\textcolor{red}{Is this correction OK?}} such that
 %---------------%
\begin{equation}
\left\{
\begin{split}
& \widehat f(a_k) = 1, \\
& \widehat f(v) = 0 \quad {\it for} \;
v \in \partial\mathcal V\setminus\big((\partial\mathcal V)_R
\cup\{a_k% a_k^{\ast}
\}\big)
\end{split}
\right.
\nonumber
\end{equation}
\HOX{ $a_k^{\ast}$ is removed.}
 %---------------%
and the partial Neumann data $\widehat g = 0$ on $(\partial\mathcal V)_L$. \\
\noindent
(2) This solution satisfies
 %---------------%
\begin{equation}
\widehat u(v) = 0  \quad \text{in vertices below (not including) the line}\;A_k.
\nonumber
\end{equation}
(3)
Moreover, the values of $\widehat f$ on $(\partial \mathcal V)_R$ and
%\HOX{\pe{Isn't $a_k^{\ast}$ a part of $\partial \mathcal V$?}\textcolor{blue}{$a_k^{\ast}$ is a part of $\partial\mathcal V$. May be it is better to say that there exist a solution $\hat u$ and a boundary data $\hat f$ ....}} 
$a_k^{\ast}$ are uniquely determined by the D-N map.
\end{lemma}
 %---------------%
\begin{proof}
The argument is the same as in Lemma 6.2 of \cite{AIM18}, which is in turn based on Lemma 6.1 of the same paper. The assertion (3) follows from
claims (2), (3) of the indicated Lemma 6.1.
\end{proof}

We determine $V_e(z)$ and $C_v$ inductively by sliding down the line $A_k$. Assuming that we know $V_e(z)$ and $C_v$ for all $e$ and $v$ above $A_k$, we determine $V_e(z)$ between $A_k$ and $A_{k-1}$,  $C_v$ for $v \in A_k$. Since the D-N map is given, by {the} induction hypothesis,
we can then compute $\widehat u(v,\lambda)$ in Lemma \ref{Lemmaspecialboundarydata} for all $v$ above $A_k$, {including this line as well,} as a meromorphic function of $\lambda$ by using the equation (\ref{S2BVP})\footnote{More detailed explanation is given in the proof of Lemma \ref{S3AboveAkLemma} {below}.}.

{Let us} compute the values of $\widehat u(v,\lambda)$ more carefully. Observing Figure \ref{ZigzagSquare},
 %---------------%
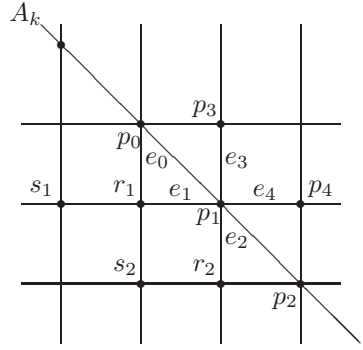
\begin{figure}[hbtp]
\begin{center}
\setlength{\unitlength}{0.53cm}
\begin{picture}(9,9)
\put(0.7,7.6){$A_k$}
%\put(-1.1,7,6){$A_{k-1}$}
\put(1,5){\line(1,0){8}}
\put(1,3){\line(1,0){8}}
\put(1,1){\line(1,0){8}}
\put(2,-0.5){\line(0,1){8}}
\put(4,-0.5){\line(0,1){8}}
\put(6,-0.5){\line(0,1){8}}
\put(8,-0.5){\line(0,1){8}}
\put(1.5,7.5){\line(1,-1){8}}
%\put(-0.5,7.5){\line(1,-1){8.5}}
\put(1.2,3.3){$s_1$}
\put(2,3){\circle*{0.2}}
\put(2,7){\circle*{0,2}}
%\put(4,7){\circle*{0,2}}
%\put(6,7){\circle*{0,2}}
%\put(8,7){\circle*{0,2}}
\put(3.3,3.3){$r_1$}
\put(4,3){\circle*{0.2}}
\put(5.4,2.6){$p_1$}
\put(6.8,3.2){$e_4$}
\put(4.7,3.2){$e_1$}
\put(6,3){\circle*{0.2}}
\put(8.2,3.3){$p_4$}
\put(8,3){\circle*{0.2}}
\put(3.4,4.5){$p_{0}$}
\put(4.1,4){$e_0$}
\put(4,5){\circle*{0.2}}
\put(3.3,1.3){$s_2$}
\put(4,1){\circle*{0.2}}

\put(5.3,5.3){$p_3$}
\put(6,5){\circle*{0.2}}
\put(6.1,2){$e_2$}
\put(6.1,4){$e_3$}
\put(5.3,1.3){$r_2$}
\put(6,1){\circle*{0.2}}

\put(7.3,0.5){$p_2$}
\put(8,1){\circle*{0.2}}

\end{picture}
\end{center}
\caption{{Diagonal} line in the square lattice}
\label{ZigzagSquare}
\end{figure}
 %---------------%
we use the equation (\ref{Equationwithcenterv0}) for the cross with center $r_1$. As $\widehat u = 0$ at $r_1, r_2, s_1, s_2$, we have
 %---------------%
\begin{equation}
\frac{1}{\phi_{e_0}(1,\lambda)}\widehat u(p_0,\lambda) +
\frac{1}{\phi_{e_1}(1,\lambda)}\widehat u(p_1,\lambda) = 0.
\label{Crosswithcenterr1}
\end{equation}
 %---------------%
Observing {next} the cross with center $p_1$, we {obtain}
 %---------------%
\begin{equation}
\frac{1}{\phi_{e_3}(1,\lambda)}\widehat u(p_3,\lambda) +
\frac{1}{\phi_{e_4}(1,\lambda)}\widehat u(p_4,\lambda) =
\left(
\sum_{i=1}^4\frac{\phi'_{e_i}(1,\lambda)}{\phi_{e_i}(1,\lambda)} + C_{p_1}\right)\widehat u(p_1,\lambda).
\label{Equation2intheSquarelattice}
\end{equation}
 %---------------%
\begin{lemma}
\label{LemmaSquareDetermineVeCv1}
(1) Assume that we know $\widehat u(v,\lambda)$ for $v = p_0, p_1$, and $V_e(z)$ for $e = e_0$. Then, we can determine $V_e(z)$ for $e = e_1$.

\noindent
(2) Assume that we know $\widehat u(v,\lambda)$ for $v = p_3, p_4, p_1$. Assume also that we know $V_e(z)$ for $e = e_3, e_4, e_1$.
Then, we can determine $V_e(z)$ for $e = e_2$, and $C_v$ for $v = p_1$.
\end{lemma}
 %---------------%
\begin{proof}
As
$\displaystyle{
\phi_{e_1}(1,\lambda) = - \frac{\widehat u(p_1,\lambda)}{\widehat u(p_0,\lambda)}\phi_{e_0}(1,\lambda)}$, one can compute the zeros of $\phi_{e_1}(1,\lambda)$ from the given data. This means that the Dirichlet eigenvalues {referring to the potential}  $V_{e_1}(z)$ are determined by {these} data. {Indeed, by} Borg's theorem (see e.g. \cite{PoTru}, p. 55 and \cite{FrYu01}, p. 27), a symmetric potential is determined by its Dirichlet eigenvalues. {This allows to construct} $V_{e_1}(z)$ uniquely {proving thus the first claim}.

Rewrite {next relation} (\ref{Equation2intheSquarelattice}) as
 %---------------%
\begin{equation}
\begin{split}
 \frac{\phi'_{e_2}(1,\lambda)}{\phi_{e_2}(1,\lambda)} + C_{p_1}  = \frac{1}{\phi_{e_3}(1,\lambda)}\frac{\widehat u(p_3,\lambda)}{\widehat u(p_1,\lambda)} + \frac{1}{\phi_{e_4}(1,\lambda)} \frac{\widehat u(p_4,\lambda)}{\widehat u(p_1,\lambda)}
- \sum_{i=1,3,4}\frac{\phi'_{e_i}(1,\lambda)}{\phi_{e_i}(1,\lambda)}.
\end{split}
\label{SquareequationcotainingCp1}
\end{equation}
 %---------------%
{Both} sides are meromorphic with respect to $\lambda \in {\bf C}$ {and the} singular points of the right-hand side are determined by the given data {only}. Noting that $\phi'_e(1,\lambda) \neq 0$ if $\phi_e(1,\lambda) = 0$ ({recall that} $\phi_e(1,\lambda) = \phi'_e(1,\lambda) = 0$ {implies} $\phi_e(z,\lambda) = 0$ for all $z$), we see that the zeros of $\phi_{e_2}(1,\lambda)$ are determined {exclusively} by the given data and that the value of $C_{p_1}$ {plays no role in fixing} the zeros of $\phi_{e_2}(1,\lambda)$. Since {these} zeros are the Dirichlet eigenvalues {referring to the potential}  $V_{e_2}(z)$, Borg's theorem {implies that} $V_{e_2}(z)$ {is determined by the data uniquely}. The equation (\ref{SquareequationcotainingCp1}) then gives $C_{p_1}$ {which} proves {the second claim}.
\end{proof}

As in Figure \ref{AkAk-1Square}, we draw {the} line, {denoted there} as $c$, passing through the upper-left corner of {the lattice} $\mathcal V^o$, and call it $A_0$, {and later also} $B_{n+1}$.

We need one more {preparatory result} before stating the following lemma. Comparing Figures~\ref{ZigzagSquare} and \ref{ZigzagSquare2}, we rewrite the equation (\ref{Crosswithcenterr1}) {above} for the cross with {the} center $a_{k-1,i+1}$ {in the form}
 %---------------%
\begin{equation}
\widehat u(a_{k,{i+1}},\lambda) = - \frac{\phi_{e'_{k,i}}(1,\lambda)}{\phi_{e_{k,i}}(1,\lambda)}\widehat u(a_{k,i},\lambda),
%\label{uak+1=fracphiuakformula}
\nonumber
\end{equation}
 %---------------%
and the equation (\ref{Equation2intheSquarelattice}) for the cross with center $a_{k,i+1}$ as
 %---------------%
\begin{eqnarray}
& &\frac{1}{\phi_{e_{k+1,i}}(1,\lambda)}\widehat u(a_{k+1,i},\lambda) +
\frac{1}{\phi_{e'_{k+1,i+1}}(1,\lambda)}\widehat u(a_{k+1,i+1},\lambda)
\nonumber
\\
&=& \left(
\frac{\phi'_{e'_{k,i}}(1,\lambda)}{\phi_{e'_{k,i}}(1,\lambda)} +
\frac{\phi'_{e_{k,i+1}}(1,\lambda)}{\phi_{e_{k,i+1}}(1,\lambda)} +
\frac{\phi'_{e_{k+1,i}}(1,\lambda)}{\phi_{e_{k+1,i}}(1,\lambda)} +
\frac{\phi'_{e'_{k+1,i+1}}(1,\lambda)}{\phi_{e'_{k+1,i+1}}(1,\lambda)} + C_{a_{k,i+1}}
\right)\widehat u(a_{k,i+1},\lambda).
\nonumber
\end{eqnarray}
 %---------------%
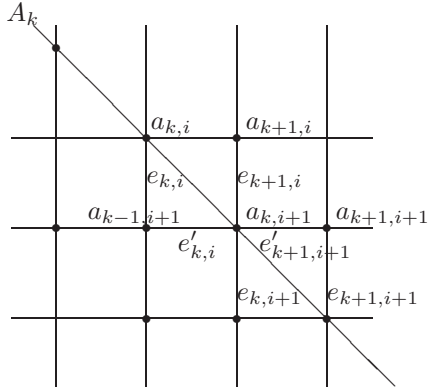
\begin{figure}[hbtp]
\begin{center}
\setlength{\unitlength}{0.6cm}
\begin{picture}(9,9)
\put(0.9,7.6){$A_k$}
\put(1,5){\line(1,0){8}}
\put(1,3){\line(1,0){8}}
\put(1,1){\line(1,0){8}}
\put(2,-0.5){\line(0,1){8}}
\put(4,-0.5){\line(0,1){8}}
\put(6,-0.5){\line(0,1){8}}
\put(8,-0.5){\line(0,1){8}}
\put(1.5,7.5){\line(1,-1){8}}

\put(2,3){\circle*{0.2}}

\put(4,3){\circle*{0.2}}
\put(8.2,3.2){$a_{k+1,i+1}$}
\put(6.2,3.2){$a_{k,i+1}$}
\put(6.2,5.2){$a_{k+1,i}$}

\put(4.7,2.5){$e'_{k,i}$}
\put(6.5, 2.5){$e'_{k+1,i+1}$}
\put(2.7,3,2){$a_{k-1,i+1}$}
\put(6,3){\circle*{0.2}}
\put(8,3){\circle*{0.2}}
\put(4.1,5.2){$a_{k,i}$}
\put(4,4){$e_{k,i}$}
\put(6,4){$e_{k+1,i}$}
%\put(2.3,4.8){$e'_{k,i-1}$}
\put(4,5){\circle*{0.2}}

\put(4,1){\circle*{0.2}}
\put(2,7){\circle*{0.2}}

\put(6,5){\circle*{0.2}}
\put(6,1.4){$e_{k,i+1}$}
\put(8,1.4){$e_{k+1,i+1}$}

\put(6,1){\circle*{0.2}}

\put(8,1){\circle*{0.2}}

\end{picture}
\end{center}
\caption{{Diagonal} line in the square lattice {again}}
\label{ZigzagSquare2}
\end{figure}
 %---------------%
\begin{lemma}
\label{LemmaSquareDetermineVeCv2}
Let $1 \leq k \leq m-1$, and assume that we know $V_e(z)$ and $C_v$ for all $e$ and $v$ above {the line} $A_k$ ({not including the points} $ v\in A_k$). Then
{one is able to} determine $V_e(z)$ and $C_v$ for all $e$ between $A_k$ and $A_{k-1}$ and $C_v$ for all $v \in A_k$.
\end{lemma}
 %---------------%
\begin{proof}
{Observe} Figure \ref{ZigzagSquare2} and {put} $i = 0$. {Considering} the cross with {the} center $a_{k-1,1}$, {one can employ the first claim of} Lemma \ref{LemmaSquareDetermineVeCv1} {to} determine $V_e(z)$ for $e = e'_{k,0}$. {Using further the second claim}, one can determine $V_e(z)$ for $e = e_{k,1}$ and $C_v$ for $v = a_{k,1}$. This is the {first} step; {proceeding then} inductively {with} $i = 1, 2, \dots$ {we prove} obtain the lemma.
\end{proof}

We have thus determined all {the} $V_e(z)$ and $C_v$ for {the edges} $e$ and {vertices} $v$ above {the line} $A_0$ {(excluding $v\in A_0)$}. {Modifying the used argument, we can deal with the part below} the line $A_0 = B_{n+1}$.

 %---------------%
\begin{lemma}
\label{LemmaSquareDetermineVeCv3}
Let $2 \leq \ell \leq n+1$, and assume that we know $V_e(z)$ and $C_v$ for all {the edges} $e$ and {vertices} $v$ above $B_{\ell}$, where $v \not\in  B_{\ell}$. Then we can determine $V_e(z)$ and $C_v$ for all {the} $e$ between $B_{\ell}$ and $B_{\ell-1}$ and $C_v$ for all $v \in B_{\ell}$.
\end{lemma}
 %---------------%
\begin{proof}
{Consider} the case $\ell = n+1$ {and denote} the points in $\mathcal V\cap B_{n+1}$  by
 %---------------%
$$
c, \ c_1, \ \dots,\  c^{\ast}.
$$
 %---------------%
{We use} Lemma \ref{Lemmaspecialboundarydata} {taking} the partial  Dirichlet data as
\begin{equation}
 \left\{
 \begin{split}
 & \widehat f(b_n) = 1, \\
 & \widehat f(v) = 0, \quad v \in \partial\mathcal V \setminus ((\partial\mathcal V)_R \cup \{c, c^{\ast}\}),
 \end{split}
 \right.
 \end{equation}
  %---------------%
and the partial Neumann data as
 %---------------%
  %\HOX{\textcolor{red}{In (2.9), $1$ should be $0$?}}
 \begin{equation}
 \left\{
 \begin{split}
 & \widehat g(b_n) = 0, \\
 & \widehat g(v) = 0, \quad v \in \partial\mathcal V \setminus ((\partial\mathcal V)_L \cup \{b_n\}),
 \end{split}
 \right.
 \end{equation}
  %---------------%
{to which the solution} $\widehat u$ constructed {there corresponds}. Then $\widehat u = 0$ {holds} below $B_{n+1}$ except {at} $b_n$, {in} fact, this is true for the vertices adjacent to $(\partial\mathcal V)_L$ {in view of} the Neumann data. Inspecting the equation, the same is true on the next line. {Arguing then as} in Lemma \ref{LemmaSquareDetermineVeCv2} {we} obtain {the claim}.
\end{proof}

We have thus proven the following theorem.
  %---------------%
\begin{theorem}
\label{TheoremTBVPspecialSquare}
Let $\Omega$ be a rectangular domain as in Figure \ref{SquareIntDomain}. {From the knowledge of} the D-N map of $- \widehat{\Delta}_{\mathcal V,\lambda} + \widehat Q_{\mathcal V,\lambda}$ for {any} energy $\lambda$, one can {then} determine all {the} $V_e(z)$ and $C_v$, provided we know $V_e(z)$ for all {the edges} $e$ adjacent to the boundary of $\mathcal V$.
\end{theorem}
  %---------------%

Theorems \ref{TheoremIBVP} and \ref{TheoremSquareISP} for the square lattice then follow from Theorem~\ref{TheoremTBVPspecialSquare}.

\section{Hexagonal lattice}
\label{SectionHexagonal}

\subsection{Hexagonal parallelogram}
Let us next consider the hexagonal lattice. The first three steps {of the construction} are parallel to the {square} case {discussed above}.
We identify $\mathbb{R}^2$ with $\mathbb{C}$, and put
 %---------------%
\begin{equation}
\omega = e^{\pi i/3}, \quad
{\bf v}_1 = 1 + \omega, \quad {\bf v}_2 = \sqrt3 i,
%\label{v1andv2inthepreviouspaper}
\nonumber
\end{equation}
 %---------------%
\begin{equation}
p_1 = \omega^{-1} = \omega^5, \quad p_2 = 1,
%\label{p1andp2inthepreviouspaper}
\nonumber
\end{equation}
 %---------------%
{Given} $n = n_1 + i n_2 \in \mathbb{Z}[i]= \mathbb{Z} + i\mathbb{Z}$, {we denote}
 %---------------%
$$
\mathcal L_0 = \left\{{\bf v}(n)\, ; \, n \in \mathbb{Z}[i]\right\}, \quad
{\bf v}(n) = n_1{\bf v}_1 + n_2{\bf v}_2,
$$
 %---------------%
and define the vertex set $\mathcal V_0$ by
 %---------------%
$$
\mathcal V_0 = \mathcal V_{01} \cup \mathcal V_{02}, \quad \mathcal V_{0i} = p_i + \mathcal L_0.
$$
 %---------------%
Let $\mathcal D_0$ be the Wigner-Seitz cell of $\mathcal V_0$. It is a hexagon the six vertices of which are at the points $\omega^k,\ 0 \leq k \leq 5$, and the center at the origin. Take the set $D_N := \{n \in \mathbb{Z}[i]\, ; \, 0 \leq n_1 \leq N, \  0 \leq n_2 \leq N\}$, and put
 %---------------%
 $$
 \mathcal D_N = {\bigcup}_{n \in D_N}\Big( \mathcal D_0 + {\bf v}(n)\Big);
 $$
 %---------------%
the number $N$ {we have to choose} large enough. This is a parallelogram in the hexagonal lattice as sketched in Figure \ref{S6HexaParallel}.
 %---------------%
\begin{figure}[h]
\includegraphics[width=7cm,clip]{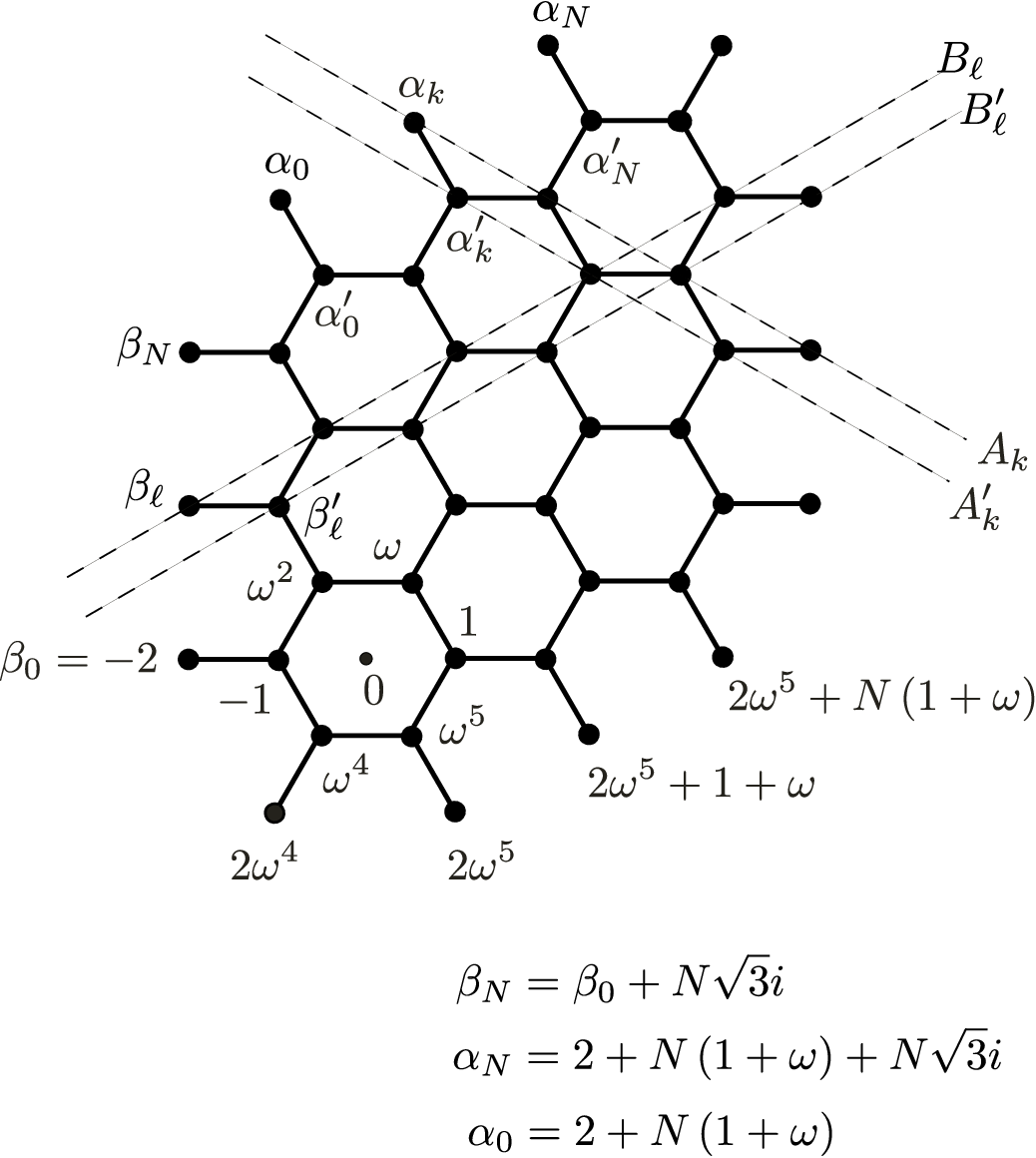}
\caption{Hexagonal parallelogram ($N = 2$)}
\label{S6HexaParallel}
\end{figure}
 %---------------%
The piecewise linear loop which is the perimeter of $\mathcal D_N$ has interior angles switching between the values $2\pi/3$ and $4\pi/3$. Let $\mathcal A$ be the set of vertices with the former {angle}, and to each $z \in \mathcal A$ we attach a new edge $ e_{z,\zeta}$ with the new vertex $\zeta = t(e_{z,\zeta})$ at its terminal point, naturally laying outside of $\mathcal D_N$. Let
 %---------------%
$$
\Omega = \{v \in \mathcal V_0\, ; \, v \in \mathcal D_N\}
$$
 %---------------%
be the set of vertices in the interior of the resulting discrete graph; its boundary consists of the added vertices, $\partial\mathcal V = \{t(e_{z,\zeta})\, ; \, z \in \mathcal A\}$, and {one is able to divide it} into four parts, naturally labeled as top, bottom, right, and left:
 %---------------%
\begin{gather*}
\begin{split}
 (\partial\mathcal V)_T =& \{ \alpha_0 , \cdots , \alpha_N \}, \\
 (\partial\mathcal V)_B = & \{2\omega^5 + k(1 + \omega)\,  ; \, 0 \leq k \leq N\}, \\
 (\partial\mathcal V)_R = & \{ 2+ N(1+\omega ) + k\sqrt{3} i \, ; \, 1 \leq k \leq N \} \cup \{ 2+N(1+\omega ) +N\sqrt{3} i + 2 \omega ^2  \} , \\
 (\partial\mathcal V)_L =& \{2\omega^4\}\cup\{\beta_0,\cdots,\beta_N\} ,
\end{split}
\end{gather*}
 %---------------%
where $ \alpha_k = \beta _N + 2\omega + k (1+\omega ) $ and $ \beta _k = -2 + k\sqrt{3} i $ for $ 0\leq k \leq N$. {The} perturbations we {are going to} consider are {again} supposed to be of compact support, {hence} taking $N$ large enough {one} can ensure that {they are supported within} $\mathcal D_N$.

For $0 \leq k \leq N$, consider the line $A_k$ as sketched in Figures \ref{S6HexaParallel} and \ref{LineAk}:
 %---------------%
\begin{figure}[h]
\centering
\includegraphics[width=6.5cm]{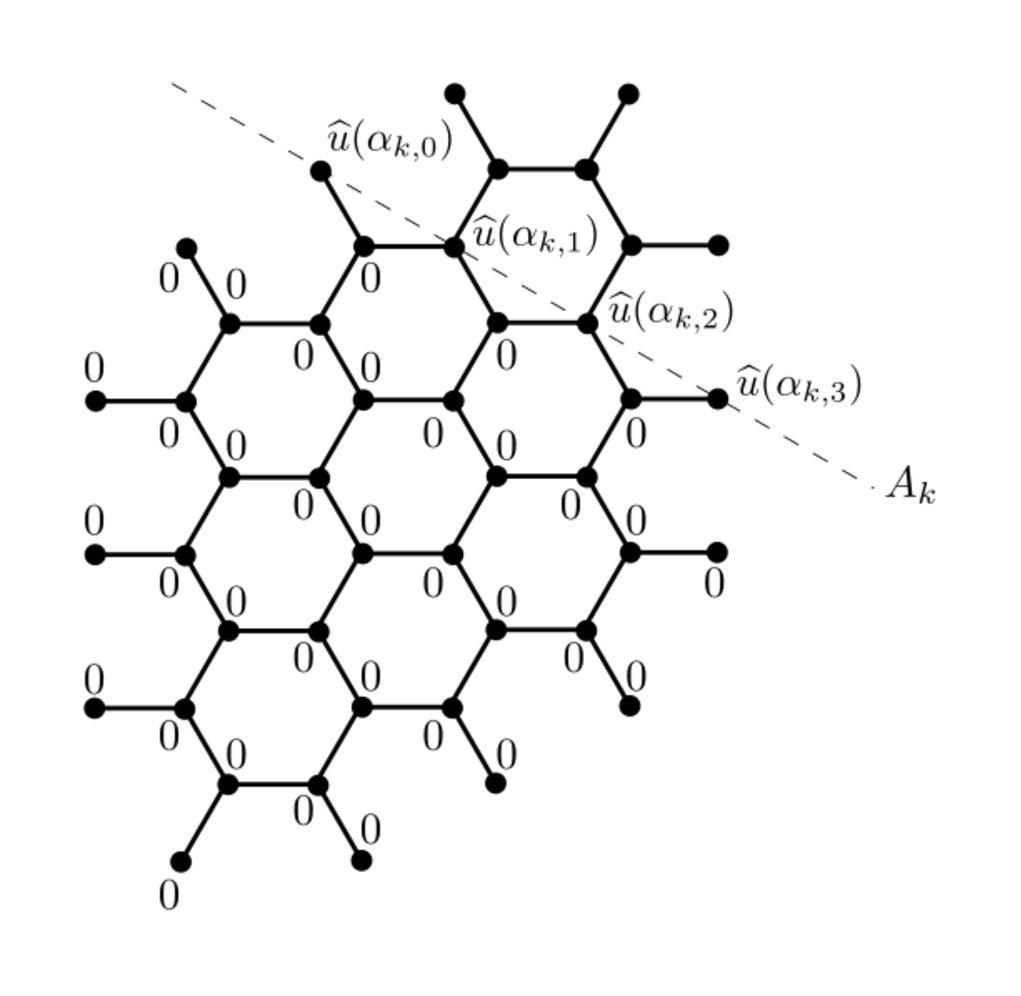}
\caption{{The line} $A_k$}
\label{LineAk}
\end{figure}
 %---------------%
\begin{equation}
A_k = \{x_1 + ix_2\, ; \, x_1 + \sqrt3x_2 = a_k\},
%\label{S6Dinagonalline}
\nonumber
\end{equation}
 %---------------%
where $a_k$ is chosen so that $A_k$ passes through
 %---------------%
\begin{equation}
\alpha_k =  \alpha_0 + k (1+ \omega )  \in
(\partial\mathcal V)_T.
\nonumber
\end{equation}
 %---------------%
The vertices {belonging to} $A_k\cap\Omega$ are written as
 %---------------%
\begin{equation}
\alpha_{k,\ell} = \alpha_k + \ell (1 + \omega^5 ), \quad \ell = 0, 1, 2, \cdots.
\nonumber
\end{equation}
 %---------------%

\subsection{Special solutions to the vertex Schr{\"o}dinger equation}
We consider the following Dirichlet problem for the discrete Schr{\"o}dinger equation
 %---------------%
\begin{equation}
\left\{
\begin{split}
& ( - \widehat{\Delta}_{\mathcal V,\lambda} + \widehat Q_{\mathcal V,\lambda})\widehat u = 0 \quad
{\rm in} \quad \mathcal V^{o}, \\
& \widehat u = \widehat f \quad {\rm on} \quad \partial \mathcal V.
\end{split}
\right.
\end{equation}
 %---------------%
As in the case of the square lattice, the following {lemmata} hold.

%%%%%%%%%%%%%%%%%%%%%% Lemma 6.1 %%%%%%%%%%%%%%%%%%%%%%%%%

\begin{lemma}\label{S6partialDNdata2}
(1) Given partial Dirichlet data $\widehat f$ on $\partial\mathcal V\setminus(\partial \mathcal V)_R$, and partial Neumann data $\widehat g$ on $(\partial\mathcal V)_L$, there is a unique solution $\widehat u$ on $\mathcal V$ to the equation
 %---------------%
\begin{equation}
\left\{
\begin{split}
& (- \widehat\Delta_{\mathcal V,\lambda} + \widehat Q_{\mathcal V,\lambda})\widehat u = 0, \quad {\it in} \quad \mathcal V^{o},\\
& \widehat u =\widehat f, \quad {\it on} \quad \partial\mathcal V\setminus(\partial\mathcal V)_R, \\
& \partial_{\nu}\widehat u = \widehat g, \quad {\it on} \quad (\partial\mathcal V)_L.
\end{split}
\right.
\label{Lemma61Equation}
\end{equation}
 %---------------%
\noindent
(2) Given the D-N map $\Lambda_{\mathcal V}(\lambda)$, partial Dirichlet data $\widehat f_2$ on $\partial\mathcal V\setminus(\partial\mathcal V)_R$ and partial Neumann data $\widehat g$ on $(\partial\mathcal V)_{L}$, there exists a unique $\widehat f$ on $\partial\mathcal V$ such that $\widehat f = \widehat f_2$ on $\partial\mathcal V\setminus(\partial\mathcal V)_R$ and $\Lambda_{\mathcal V}(\lambda)\widehat f = \widehat g$ on $(\partial\mathcal V)_{L}$. Moreover, $\widehat f$ is uniquely determined by the D-N map.
\end{lemma}

 %---------------%
\begin{lemma}
\label{LemmaspecialboundarydataHexa2}
Let $A_k \cap\partial\mathcal V = \{\alpha_{k,0},\alpha_{k,m}\}$. \\
\noindent
(1) There exists a unique solution $\widehat u$ to the equation
 %---------------%
\begin{equation}
\big(- \widehat\Delta_{\mathcal V,\lambda} + \widehat Q_{\mathcal V,\lambda}\big)\widehat u = 0 \quad {\it in} \quad \mathcal V^{o},
\label{S7BVP}
\end{equation}
 %---------------%
with the partial Dirichlet data $\widehat f$ such that
 %---------------%
\begin{equation}
\left\{
\begin{split}
& \widehat f(\alpha_{k,0}) = 1, \\
& \widehat f(z) = 0 \quad {\it for} \quad
z \in \partial\mathcal V\setminus\big((\partial\mathcal V)_R
\cup\alpha_{k,0}\cup%\alpha_{k,m}
\big)
\end{split}
\right.
\nonumber
\end{equation}
\HOX{As in Lemma 2.2, $\alpha_{k,m}$ is removed.}
 %---------------%
and the partial Neumann data $\widehat g = 0$ on $(\partial\mathcal V)_L$. \\
\noindent
(2) \pe{The said solution} satisfies
 %---------------%
\begin{equation}
\widehat u(x_1 + ix_2) = 0 \quad {\it if} \quad
x_1 + \sqrt3 x_2 < a_k.
%\label{S7Lemmu=0belowAk}
\nonumber
\end{equation}
(3)
Moreover, the values of $\widehat f$ on $(\partial \mathcal V)_R$ and $\alpha_{k,m}$ are uniquely determined by the D-N map.
\end{lemma}
 %---------------%
\noindent

 %---------------%
\begin{lemma}
\label{S3AboveAkLemma}
\noindent
The values of $\widehat u$ on $A_k$ can be computed by {only using} the D-N map, the edge potentials $V_e(z)$ for the edges $e$ in the halfplane $x_1 + \sqrt{3}x_2 \geq a_k$, and the vertex potentials $C_v$ for vertices $v$ in the halfplane $x_1 + \sqrt{3}x_2 > a_k$.
\end{lemma}
 %---------------%
\begin{proof}
The claim is obtained using equation (\ref{S7BVP}) and the D-N map. Specifically, we regard the said equation as the Cauchy problem with initial data on the top and left sides of $\mathcal V$. More precisely, we argue in the same way as in \cite{AIM18}, Lemma 6.1. Let us inspect Figure \ref{S6HexaParallel}. Given the D-N map, one can compute the values of $\widehat u(v)$ at the vertices $v$ on the line $x_1 = -1$.  They are zero in the halfplane $x_1 + \sqrt{3}x_2 < a_k$, since the same is true for relevant boundary data. Next we observe the values of $\widehat u(v)$ on the {\it adjacent line,} $x_1 = - 1/2$. We start from $\omega^4$ and go up. Then, by using the D-N map for $\omega^4$ and then the equation (\ref{S7BVP}) for the vertices above $\omega^4$, we see that $\widehat u(v) = 0$ holds for $v$ in the halfplane $x_1 + \sqrt{3}x_2 < a_k$ when $x_1 = -1/2$. Repeating this procedure, we conclude that $\widehat u(v) = 0$ holds in the lower halfplane $x_1 + \sqrt{3}x_2 < a_k$. To compute the values of $\widehat u(v)$ in the upper halfplane, $x_1 + \sqrt{3}x_2 \geq a_k$, we observe the top
 boundary $(\partial\mathcal V)_T$. We first compute $\widehat u(v)$ for vertices $v$ adjacent to $(\partial\mathcal V)_T$. They are determined by the D-N map. Next we observe $\widehat u(v)$ on the neighboring line, the second adjacent to the boundary. This time, we start from the rightmost vertex (for which, by virtue of claim (3) from Lemma \ref{Lemmaspecialboundarydata}, the value $\widehat u(v)$ is determined by the D-N map) and go down to the left. With the aid of the equation, we can then determine $\widehat u(v)$ by using the values of $V_e(z)$ and $C_v$ in the halfplane $x_1 + \sqrt{3}x_k > a_k$, which are known by assumption. Finally, using the equation, the values of $\widehat u(v)$ on the line $x_1 + \sqrt{3}x_2 = a_k$ are computed from those in the halfplane $x_1 + \sqrt{3}x_2 > a_k$ and $V_e(z), C_v$ specified in the lemma.
\end{proof}

Now we pass to the reconstruction procedure.

\smallskip
\noindent
Let $\widehat u$ be the solution to the equation (\ref{S7BVP}) specified in Lemma \ref{LemmaspecialboundarydataHexa2}. If $a_k$ is large enough, all the perturbations lie below the line $A_k$, hence the support of $\widehat u$ is disjoint with them. Let us slide $A_k$ downwards, and observe the equation when we first touch the perturbation. Let $a, b, b', c \in \mathcal V$ and ${e}, {e}' \in \mathcal E$ be as in  Figure~\ref{S7hex_edge}.
 %---------------%
\begin{figure}[h]
\includegraphics[width=5cm]{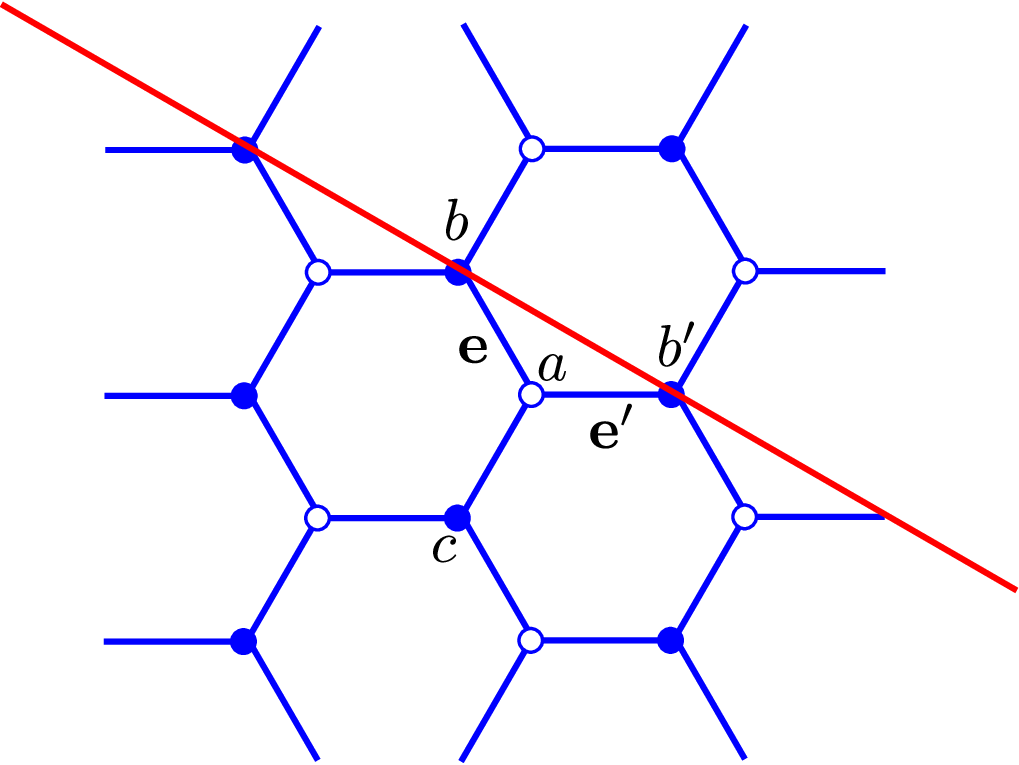}
\caption{{The values} $\widehat u(b)$ and $\widehat u(b')$ {considered}}
\label{S7hex_edge}
\end{figure}
 %---------------%
Then, evaluating the equation (\ref{S7BVP}) at $v = a$, we obtain
 %---------------%
\begin{equation}
\widehat u(b) = - \frac{\phi_{e0}(1,\lambda)}{\phi_{e'0}(1,\lambda)}\,\widehat u(b').
%\label{S7ubandub'}
\nonumber
\end{equation}
 %---------------%
Let ${e}_{k,1}, {e}'_{k,1}, {e}_{k,2}, {e}'_{k,2}, \cdots$ be the series of edges just below $A_k$ starting from the vertex $\alpha_k$, and put
 %---------------%
\begin{equation}
f_{k,m}(\lambda) = - \frac{\phi_{e_{k,m}0}(1,\lambda)}{\phi_{e'_{k,m}0}(1,\lambda)}.
\label{s7equationforfkm}
\end{equation}
 %---------------%
Then we see that the  solution $\widehat u$ in  Lemma \ref{LemmaspecialboundarydataHexa2} satisfies
 %---------------%
\begin{equation}
\widehat u(\alpha_{k,\ell}) = f_{k,1}(\lambda)\cdots f_{k,\ell}(\lambda)\,;
\label{S3widehatuonAk}
\end{equation}
 %---------------%
note that the right-hand side does not depend on $C_v$.
 %---------------%
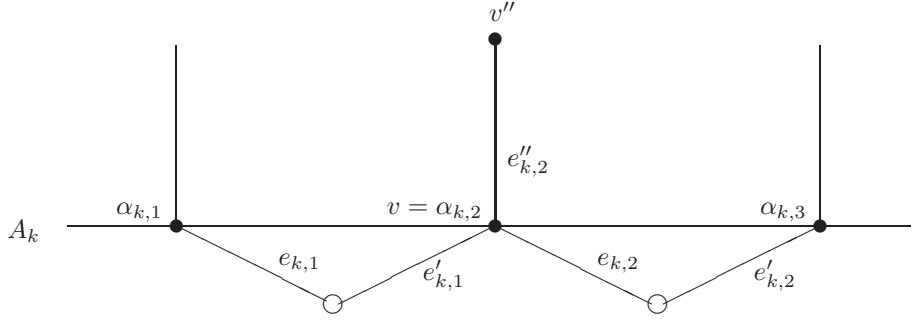
\begin{figure}[hbtp]
\begin{center}
\setlength{\unitlength}{0.8cm}
\begin{picture}(9,6)(0,-1)
\put(-2.8,0.8){$A_k$}
\put(-1.8,1){\line(1,0){14}}
\put(0,1){\circle*{0.2}}
\put(0,1){\line(0,1){3}}
\put(0,1){\line(2,-1){2.5}}
\put(-1, 1.2){$\alpha_{k,1}$}
\put(1.7,0.35){$e_{k,1}$}
\put(2.6,-0.3){\circle{0.3}}
\put(2.7,-0.3){\line(2,1){2.5}}
\put(4.1,0.1){$e'_{k,1}$}
\put(5.3,1){\circle*{0,2}}
\put(3.5,1.2){$v = \alpha_{k,2}$}
\put(5.5,2){$e''_{k,2}$}
\put(5.3,1){\line(0,1){3}}
\put(5.2,4.4){$v''$}
\put(5.3,4.1){\circle*{0,2}}
\put(5.4,0.95){\line(2,-1){2,5}}
\put(7,0.35){$e_{k,2}$}
\put(8.0,-0.3){\circle{0.3}}
\put(8.1,-0.3){\line(2,1){2.5}}
\put(9.6,0.1){$e'_{k,2}$}
\put(10.7,1){\circle*{0,2}}
\put(10.7,1){\line(0,1){3}}
\put(9.7,1.2){$\alpha_{k,3}$}
\end{picture}
\end{center}
\caption{From $e_{k,1}$ to $e'_{k,1}$, $e_{k,2}$}
\label{e1e2e3}
\end{figure}
 %---------------%

\noindent
To proceed, observe the graph detail sketched in Figure \ref{e1e2e3}.
 %---------------%
\begin{lemma}
\label{S3LemmaDetermineVe(z)Cv}
As in Lemma \ref{S3AboveAkLemma}, assume that we know $V_e(z)$ and $C_v$ for all edges $e$ and vertices $v$ above $A_k$. Then, assuming that we know $V_e(z)$ for $e = e_{k,1}$, {one} can determine $V_e(z)$ for $e = e'_{k,1}$ and $e = e_{k,2}$, and also $C_v$ for $v = \alpha_{k,2}$.
\end{lemma}
 %---------------%
\begin{proof}
In view of (\ref{S3widehatuonAk}), since we know $\widehat u(\alpha_{k,\ell})$ for all $\ell$, we also know the corresponding $f_{k,\ell}(\lambda)$. Furthermore, since we know $\phi_{e_{k,1}}(1,\lambda)$, we get from (\ref{s7equationforfkm}) $\phi_{e'_{k,1}0}(1,\lambda)$ for all values of the energy $\lambda$. The zeros of $\phi_{e'_{k,1}0}(1,\lambda)$ are the Dirichlet eigenvalues of $- (d/dz)^2 + V_e(z)$, $e = e'_{k,1}$. By  Borg's theorem, these eigenvalues determine the potential $V_{e}(z)$ at the edge $e = e'_{k,1}$.

Making the equation (\ref{S7BVP}) explicit at the vertex $v = \alpha_{k,2}$ and noting that $\widehat u = 0$ holds below $A_k$, we have
 %---------------%
\begin{equation}
\frac{\phi'_{e_{k,2}0}(1,\lambda)}{\phi_{e_{k,2}0}(1,\lambda)}  = \frac{1}{\phi_{e''_{k,2}0}(1,\lambda)}
\frac{\widehat u(v'')}{\widehat u(v)}  - \frac{\phi'_{e''_{k,2}0}(1,\lambda)}{\phi_{e''_{k,2}0}(1,\lambda)} - \frac{\phi'_{e'_{k,1}0}(1,\lambda)}{\phi_{e'_{k,1}0}(1,\lambda)} - C_v,
\label{S3DetermineVe(z)andCv}
\end{equation}
 %---------------%
where $v = \alpha_{k,2} =  e''_{k,2}(0)$, $v'' = e''_{k,2}(1)$. The both sides of (\ref{S3DetermineVe(z)andCv}) are meromorphic function in the complex plane, and we note that $\widehat u$ is also meromorphic with respect to $\lambda$ by construction. The singularities of the left-hand side are thus determined by those of the right-hand side. Furthermore, the singularities of the right-hand side can be computed without using $C_v$. Therefore, the zeros of $\phi_{e_{k,2}0}(1,\lambda)$ are  determined by $V_e(z)$ for edges $e$ in the halfplane $x_1 + \sqrt{3}x_2 \geq a_k$ and $e = e_{k,1},  e'_{k,1}$ as well as $C_w$ for $w$ in the open halfplane $x_1 + \sqrt{3}x_2 > a_k$. As before, Borg's theorem determines $V_{e_{k,2}}(z)$, which in turn {determines} $\phi_{e_{k,2}0}(1,\lambda)$ and $\phi'_{e_{k,2}0}(1,\lambda)$. The value of $C_v$ is then {computed} from the formula (\ref{S3DetermineVe(z)andCv}) for regular points $\lambda$.
\end{proof}

Lemma \ref{S3LemmaDetermineVe(z)Cv} implies the following {result}.

\begin{lemma}
\label{S5howtodetermineVez}
Let $\widehat u$ be as in Lemma \ref{LemmaspecialboundarydataHexa2}, and assume that we know $V_e(z)$ for all {the} edges $e$ above $A_k$, and $C_v$ for all {the} $v$ above $A_k$ but not on $A_k$. Then, {one} can determine all {the} $V_e(z)$ for $e$ between $A_k$ and $A'_k$, and $C_v$ for all {the} $v$ on $A_k$.
\end{lemma}
 %---------------%
\begin{figure}[hbtp]
\centering
\includegraphics[width=5cm]{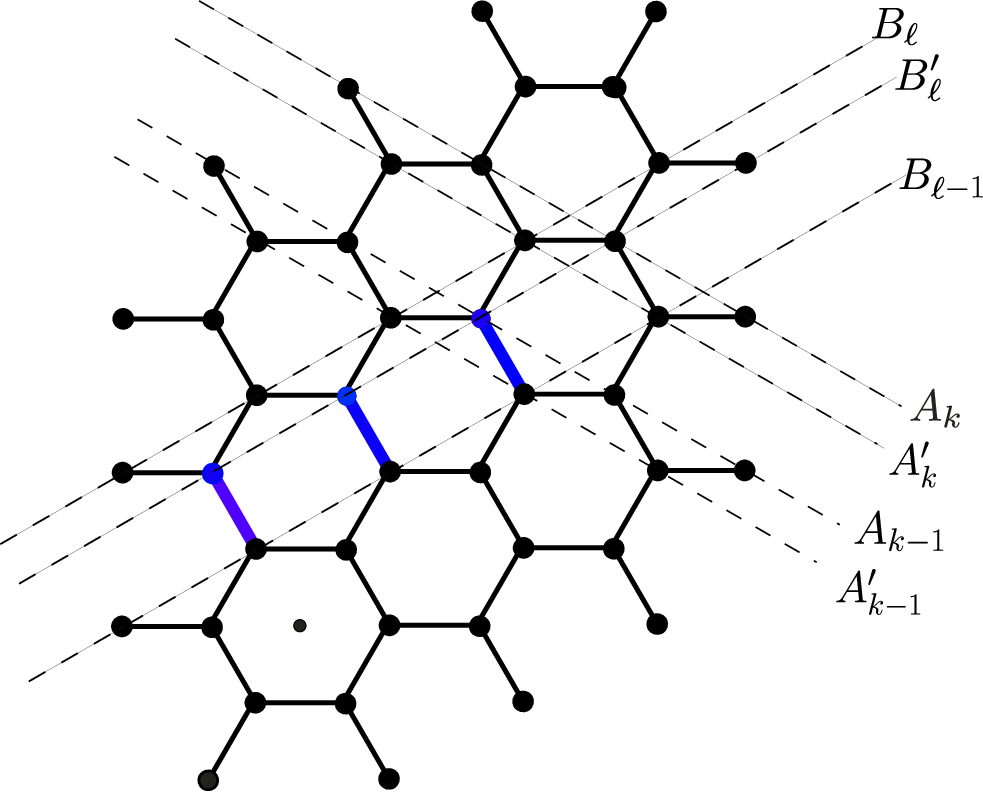}
\caption{From $B_l$ to $B_{l-1}$.}
\label{AkBl}
\end{figure}
 %---------------%

Let  $A_k,B_l$ be as in Figure \ref{AkBl}.  By Lemma \ref{S5howtodetermineVez}, we can compute (from {the knowledge of} the D-N map) the edge {potentials} between $A_k,A_k'$ and between $B_l,B_l'$, and also compute $C_v$ for {the} vertices on $A_k,B_l$. We call this \emph{the initial procedure}.

Next we {make a step} inside {$\mathcal{D}_N$}. For the notational reason, we determine the edge potentials below $B'_{\ell}$. To be able to {repeat the described} procedure {after passing to the neighboring ``lower''} line $B_{l-1}$, it suffices to know the edge potential for the edges below $B'_{\ell}$ and $A'_k$ with {the} endpoints on $B'_{\ell}$, {as well as} $C_v$ for such endpoints on $B'_{\ell}$ in Figure \ref{AkBl}.

\smallskip

The following ``local'' version of Lemma \ref{S5howtodetermineVez} holds.
 %---------------%
\begin{lemma}\label{Ak-1_local}
In Figure \ref{AkBl_local}, if one knows the edge potentials for the edges $ab$ and $ac$,  one can compute the edge potential of the edge $ac'$ and the vertex potential $C_v$ for $v=a$.
\end{lemma}
 %---------------%
\begin{figure}[hbtp]
\centering
\includegraphics[width=6.5cm]{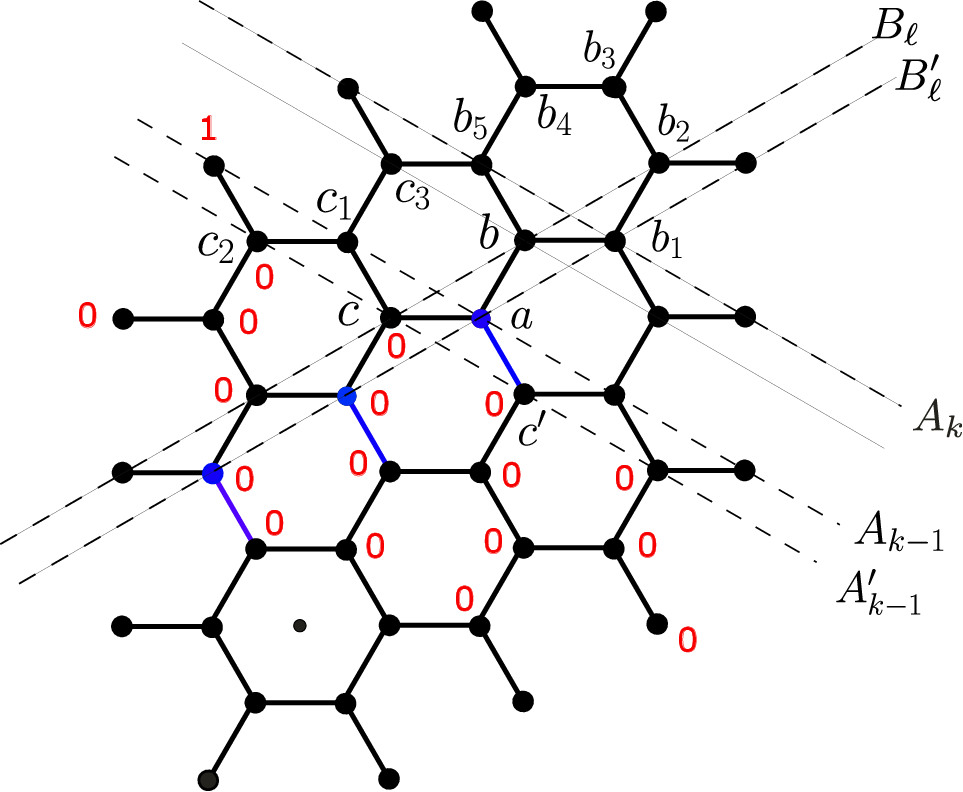}
\caption{Solving the edge potential of $ac'$ and $C_v$ at $v=a$.}
\label{AkBl_local}
\end{figure}
 %---------------%
\begin{proof}
{As in} Lemma 3.5, {consider} the equation for $A_{k-1}$, as shown in Figure \ref{AkBl_local}; {for the sake of clarity and brevity, we denote it} by $\widehat{u}:=\widehat{u}_{A_{k-1}}$. {Using the edge identification} $e=ab$, {we also use} the {symbol} $\phi_{ab}(1,\lambda)$ {for} $\phi_{e0}(1,\lambda)$.

{The} procedure to {``to} compute $\phi$'' means {to express it using the} known quantities {in the successive steps:}

\medskip
\noindent
(i) Evaluating the equation (\ref{S7BVP}) at $b_2$, we can compute $\widehat{u}_{A_{k-1}}(b_1)$.
\\
(ii) Evaluating the equation at $b_4$, we can compute $\widehat{u}_{A_{k-1}}(b_5)$. \\
(iii) Evaluating the equation at $b_5$, we can compute $\widehat{u}_{A_{k-1}}(b)$. \\
(iv) The potential {at the edges} $bb_1,ba$ and the vertex potential $C_v$ at $b$ are {found} during the initial procedure for the {lines $B$}.
Hence {solving} the equation at $b$, we can compute $\widehat{u}_{A_{k-1}}(a)$. \\
(v) The potentials {at the edges} $ab,ac$ are {found} during the initial procedure for the {lines $B$}. Since $\widehat{u}_{A_{k-1}}(c)=\widehat{u}_{A_{k-1}}(c')=0$ and $\widehat{u}_{A_{k-1}}$ at $a,b$ are known, {one} can determine the {potential at} $ac'$ and $C_v$ at $v=a$ by the same argument as {in} Lemma~3.7.

Let us explain the step (v) in detail. We consider the equation for $\widehat{u}$ at vertex $a$:
 %---------------%
$$(-\widehat{\Delta}_{\mathcal{V},\lambda}+\widehat{Q}_{a,\lambda})\widehat{u}(a)=0,$$
 %---------------%
where
 %---------------%
\begin{equation*}
(\widehat{\Delta}_{\mathcal{V},\lambda} u)(a)=\frac{1}{3} \sum_{w\sim a} \frac{1}{\phi_{aw}(1,\lambda)} \widehat{u}(w)=\frac{1}{3} \frac{1}{\phi_{ab}(1,\lambda)} \widehat{u}(b),
\end{equation*}
 %---------------%
since $\widehat{u}(c)=\widehat{u}(c')=0$. By (\ref{QVlambdaformula}),
 %---------------%
\begin{eqnarray*}
\widehat{Q}_{a,\lambda}&=&\frac{1}{3} \sum_{w\sim a} \frac{\phi_{aw}'(1,\lambda)}{\phi_{aw}(1,\lambda)}+\frac{C_a}{3} \\
&=& \frac{1}{3}\left(\frac{\phi_{ab}'(1,\lambda)}{\phi_{ab}(1,\lambda)}+\frac{\phi_{ac}'(1,\lambda)}{\phi_{ac}(1,\lambda)}+\frac{\phi_{ac'}'(1,\lambda)}{\phi_{ac'}(1,\lambda)}+C_a\right).
\end{eqnarray*}
 %---------------%
This equation yields
 %---------------%
\begin{equation} \label{eq-ac'}
\frac{\phi_{ac'}'(1,\lambda)}{\phi_{ac'}(1,\lambda)}= \frac{1}{\phi_{ab}(1,\lambda)} \frac{\widehat{u}(b)}{\widehat{u}(a)}-\frac{\phi_{ab}'(1,\lambda)}{\phi_{ab}(1,\lambda)}-\frac{\phi_{ac}'(1,\lambda)}{\phi_{ac}(1,\lambda)}-C_a.
\end{equation}
 %---------------%
By the above {construction}, $\widehat u(a)$ is a meromorhic function of $\lambda$. All {the} quantities on the right-hand side of \eqref{eq-ac'} except $C_a$ are known. This means the singularities of the right-hand side of \eqref{eq-ac'} are known. Hence we can determine the singularities of the left-hand side of
\eqref{eq-ac'} without knowing $C_a$. The singularities of the left-hand side are the Dirichlet eigenvalues for the edge $ac'$, which {determine} the edge potential of $ac'$ by Borg's theorem. This shows $\phi_{ac'}(z,\lambda),\ z\in [0,1]$ can be computed. Inserting $\phi_{ac'}(1,\lambda),\phi'_{ac'}(1,\lambda)$ back into \eqref{eq-ac'} gives $C_a$.
\end{proof}

Note that in the above proof, we use only the knowledge of edges between $B_{\ell}$ and $B_{\ell'}$. We emphasize that the knowledge of edges below both of $A_k$ and $B_{\ell}$, which is still unknown, is not used. Then we repeat Lemma \ref{Ak-1_local} by taking the function $\widehat{u}$ for $A_{k-2},A_{k-3},\cdots$, and thus we recover the edge potential for all the  edges just below $B'_{\ell}$ and left to $A_{k-1}$,
%\HOX{\textcolor{red}{Is this explanation understandable? Or we should modify the figure?}}
and $C_v$ for all  vertices between $B_{\ell}$ and $B_{\ell}'$ in Figure \ref{AkBl}. This shows that we can push the lines $B_l$ {as ``low''} as needed.

By a symmetric {reasoning,} choosing the function $\widehat{u}$ in {the} different direction (i.e., choosing $\widehat{u}_{B_j}$ for {the} proper $j$), one can show that the lines $A_k$ can also be pushed {``down'' as low} as needed. But {in fact}, this is not {needed}, since pushing the lines $B_l$ {``down''} alone can already recover the whole perturbed region.

We {arrive thus at the following conclusions:}
 %---------------%
\begin{theorem}
\label{TheoremTBVPspecialHexa}
Let $\Omega$ be a hexagonal paralellogram as in Figure \ref{S6HexaParallel}. Then from the D-N map of $- \widehat{\Delta}_{\mathcal,\lambda} + \widehat Q_{\mathcal V,\lambda}$, one can determine all {the} $V_e(z)$ and $C_v$, provided we know $V_e(z)$ for all {the} edge $e$ adjacent to the boundary of $\mathcal V$.
\end{theorem}

Theorems \ref{TheoremIBVP} and \ref{TheoremSquareISP} for the hexagonal lattice then follow from Theorem \ref{TheoremTBVPspecialHexa}.

 \end{document}